\documentclass[11pt,a4paper]{article}
\usepackage[utf8]{inputenc}
\usepackage{amsmath}
\usepackage{amsfonts}
\usepackage{amssymb}
\author{Eleni Bakali \\ National Technical University of Athens, CS dept.}
\title{On randomized counting versus randomised decision}

\usepackage{hyperref}

\usepackage{tikz}

\newtheorem{theorem}{Theorem}
\newtheorem{lemma}{Lemma}
\newtheorem{definition}{Definition}
\newtheorem{corollary}{Corollary}

\newenvironment{proof}{\noindent\textit{Proof}. }{\hfill $\Box$}

\begin{document}
\maketitle

\begin{abstract}
We study the question of which counting problems admit f.p.r.a.s., under a structural complexity perspective. Since problems in \#P with NP-complete decision version do not admit f.p.r.a.s. (unless NP = RP), we study subclasses of \#P, having decision version either in P or in RP. We explore inclusions between these subclasses and we present all possible worlds with respect to NP v.s. RP and RP v.s. P.
\end{abstract}


\section{Introduction}
The motivation for this paper is to obtain a better understanding of randomized approximate counting  inside \#P.

The class \#P \cite{Valiant79} is the class of functions that count the number of solutions to problems in NP, e.g. \#SAT is the function that on input a formula $\phi$ returns the number of satisfying assignments of $\phi.$ Equivalently, functions in \#P count accepting paths of non-deterministic polynomial time Turing machines (NPTMs).

 NP-complete problems are hard to count, but it is not the case that problems in P are easy to count as well. When we consider counting, non-trivial facts hold. First of all there exist \#P-complete problems, that have decision version in P, e.g. \#DNF. Moreover, some of them can be approximated, e.g. the Permanent \cite{JS96perm} and \#DNF   \cite{KLM89}, while other can not, e.g. \#IS \cite{DFJ02}. 

Regarding approximability, counting versions of NP-complete problems do not admit f.p.r.a.s. unless NP = RP. So we have to concentrate on problems whose decision version is not NP-complete. There are two subclasses of \#P that concern our study, because they have easy decision version; TotP and \#PE \cite{KPSZ01}. \#PE contains all counting problems in \#P with decision version in P, and TotP contains functions that count all paths of NPTMs, and it is also proven to contain all self-reducible problems in \#PE \cite{PZ06}. 

It is observed that all problems in \#P known to admit f.p.r.a.s., belong to TotP.   This observation lead to some questions. Which counting problems admit f.p.r.a.s. from a structural complexity perspective? Does FPRAS $\subseteq$ TotP hold? (In this paper we denote  FPRAS the class of problems in \#P that admit an f.p.r.a.s.) Is there an alternative characterization of FPRAS, in terms of counting paths of NPTMs, like the respective definitions of \#P and TotP? 

In this article we introduce two counting classes, \#RP1 and \#RP2, that can both be considered as counting versions of problems with decision version in RP. We show that FPRAS  lies between them. We also show that  these two classes are not equal, unless NP=RP. 

Moreover we show that if FPRAS $\subseteq$ TotP, then RP=P. So we don't resolve the question of FPRAS vs TotP, but it turns out that if the above inclusion holds, then it is more difficult to prove it, than proving RP=P. Note that for the opposite direction, i.e. TotP $\subseteq$ FPRAS, it is known that it holds iff NP=RP.

We further explore inclusions between all the above mentioned classes, i.e. \#P, \#PE, \#RP1, \#RP2, TotP and FPRAS, and we present all possible worlds with respect to NP vs RP, and RP vs P. 

We present our results in section 3 and we discuss them in section 5. The rest of the paper is organized as follows. Section 2 contains definitions, lemmas and known results used in our proofs. Section 4 contains our proofs. Section 6 provides briefly some related work, and conclusions are given in section 7. 


\section{Preleminaries}

Let $M$ be an NPTM. If $M$ returns $yes$ or $no,$ we denote $acc_M(x)=\#$(accepting paths of $M$ on input $x$). We denote  $p_M$ the polynomial bounding the number of non-deterministic choices made by $M.$

\begin{definition} \cite{Valiant79}
$\#P$ is the class of functions $f$ for which there exists a polynomial-time decidable binary relation $R$ and a polynomial $p$ 
such that $\forall x\in\Sigma^*$,  $f(x)=\big|\{y \in \{0,1\}^*,\  s.t.  |y|=p(|x|)\wedge R(x,y)\}\big|$. 
Equivalently, $\#P=\{acc_M:\Sigma^*\rightarrow \mathbb{N} \ | \  M\text{ is a NPTM}\}$
\end{definition}

\begin{definition} FP is the class of functions in \#P computable in polynomial time.
\end{definition}

\begin{definition}
The \textbf{decision version} of a function in \#P is the language  $L_f=\{x\in \Sigma^* | f(x)>0\}$.
\end{definition}

\begin{definition} \cite{Pagourtzis01}
	$\#PE=\{f: \Sigma^*\rightarrow \mathbb{N} \ | \ f\in\#P \text{ and } L_f \in P \}$.
\end{definition} 

\begin{definition} \cite{KPSZ01}
$\text{TotP} =\{tot_M:\Sigma^*\rightarrow \mathbb{N} \ | \  M\text{ is a NPTM}\}$, where 
$tot_M(x)= \#($all computation paths of $M$ on input $x)-1$.
\end{definition}

TotP contains all self-reducible problems in \#P, with decision version in P. 

Intuitively, self-reducibility means that counting the number of solutions to an instance of a problem, can be performed recursively by computing the number of solutions of some other instances of the same problem. For example \#SAT is self reducible: the number of satisfying assignments of a formula $\phi$ is equal to  the sum of the number of satisfying assignments of $\phi_1$ and $\phi_0$, where $\phi_i$ is $\phi$ with its first variable fixed to $i$. 

The most general notion of self reducibility, that TotP can capture, is
given in \cite{PZ06}. The following is known.

\begin{theorem}\cite{PZ06}
\label{PZmt}
(a) FP $\subseteq$ TotP $\subseteq \#PE \subseteq \#P$. The inclusions are proper unless $P=NP$.\\
(b) TotP is the Karp closure of self-reducible $\#PE$ functions.
\end{theorem}

We consider FPRAS to be the class of functions in \#P that admit f.p.r.a.s.

\begin{definition} Let $f\in \#P$. $f\in FPRAS$ if there exists a randomized algorithm that on input $x \in \Sigma^*,$ $\epsilon>0,$ $\delta>0,$  returns a value $\hat{f(x)}$ s.t.
\[\Pr[(1-\epsilon) f(x) \leq \hat{f(x)} \leq (1+\epsilon) f(x)]\geq 1-\delta\] in time poly($|x|,\epsilon^{-1}, \log \delta^{-1}$).
\end{definition}

\begin{theorem}\label{TotP-vs FPRAS 1}
(a) If NP$\neq$RP then there are functions in TotP, that are not in FPRAS. \\
(b) If NP=RP then $\#P\subseteq FPRAS$. 
\end{theorem}
\begin{proof}
\#IS (i.e. given a graph, count all independent sets, of all sizes) belongs to TotP, and does not admit f.p.r.a.s. unless NP=RP \cite{DFJ02}.
The second assertion derives from a Stockmeyer's well known theorem \cite{Stockmeyer85a}, see Appendix for a proof sketch.
\end{proof}

Now we introduce two classes, having decision version in RP.
For that, we need to define the set ${\cal MR}$ of Turing Machines associated to problems in RP. 

\begin{definition}
Let $M$ be an NPTM. We denote $p_M$ the polynomial s.t.  on inputs of size $n$, $M$ makes $p_M(n)$ non-deterministic choices.\\
${\cal MR} = \{$NPTM $M\ |\ \forall x\in\Sigma^*  $ either $acc_M(x)=0$ or $acc_M > \frac{1}{2} 2^{p_M(|x|)}\}.$  
\end{definition}

\begin{definition}
\#RP1 = $\{f\in \#P\ |\ \exists M \in {\cal MR}$ s.t. $\forall x$ $f(x) = acc_M(x)\}.$
\end{definition}

\begin{definition}
\#RP2 = $\{f \in \#P \ |\ L_f \in RP \}.$
\end{definition}

For our proofs we will need the following lemma.
\begin{lemma}\label{unbiased estimator}
Unbiased estimator. Let $A\subseteq B$ be two finite sets, and let $p=\frac{|A|}{|B|}$. Suppose we take $m$ samples from $B$ uniformly at random, and let a be the number of them that belong to $A$. Then $\hat{p}=\frac{a}{m}$ is an unbiased estimator of $p$, and it suffices $m=poly(p^{-1},\epsilon^{-1},\log \delta^{-1})$ in order to have 
\[\Pr[ (1-\epsilon) p \leq  \hat{p}\leq (1+\epsilon) p] \geq 1-\delta.\]   
\end{lemma}

\section{Results}
We first present unconditional inclusions  between the above mentioned classes. These inclusions together with those already known by theorem \ref{PZmt} can be summarised in the following diagram. 

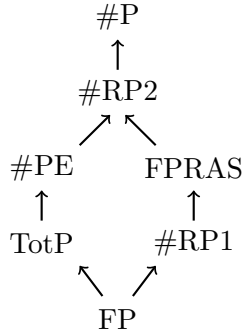
\begin{center}
\begin{tikzpicture}
\node at (1,3) {\#P};
\draw [thick, -> ] (1,2.3) -- (1,2.7);
\node at (1,2) {\#RP2};
\draw [thick, ->] (0.5,1.3) -- (0.9,1.7);
\node at (0,1) {\#PE};
\draw [thick, ->] (1.5,1.3) -- (1.1,1.7);
\node at (2,1) {FPRAS};
\draw [thick, ->] (0,0.3) -- (0,0.7);
\node at (0,0) {TotP};
\draw [thick, ->] (2,0.3) -- (2,0.7);
\node at (2,0) {\#RP1};
\node at (1,-1) {FP};
\draw [thick, ->] (0.8,-0.7) -- (0.5,-0.3);
\draw [thick, ->] (1.2,-0.7) -- (1.5,-0.3);
\end{tikzpicture} 
\end{center} 

Then we explore more inclusions between these counting classes, with respect to NP vs RP and RP vs P. The results can be summarized in the figures of theorem \ref{four worlds}. 

The most important corollaries  of this investigation are that FPRAS is not all contained in TotP unless RP=P, and \#RP1 does not coincide with \#RP2 unless NP=RP.

\section{Proofs}
\subsection{Unconditional inclusions}

\begin{theorem}
$FP\subseteq \#RP1 \subseteq \#RP2 \subseteq \#P.$ Also $\#PE\subseteq \#RP2$
\end{theorem}
\begin{proof}
Let $f\in FP$. We will show $f\in \#RP1$. We will construct an NPTM $M\in {\cal MR}$ s.t. on input $x$ $acc_M(x)=f(x)$. Let $x\in \Sigma^*$. $M$ computes $f(x)$. Then it computes $i\in \mathbb{N}$ s.t. $f(x)\in (2^{i-1},2^i]$. $M$ makes $i$ non-deterministic choices $b_1,b_2,...,b_i$. Each such $b\in \{0,1\}^i$ determines a path, in particular, $b$ corresponds to the $(b+1)$ path (since $0^i$ is the first path). $M$ returns yes iff $b+1 \leq f(x)$, so $acc_M=f(x)$. Since $f(x) > 2^{i-1},$ $M\in{\cal MR}.$ 

The other inclusions are immediate by definitions.
\end{proof}

\begin{theorem}\label{main theorem}
$\#RP1 \subseteq FPRAS \subseteq \#RP2.$
\end{theorem}
\begin{proof}
For the first inclusion, let $\epsilon > 0, \delta >0.$ Let $f\in \#RP1$. There exists a $M_f\in {\cal MR}$ s.t. $\forall x$ $acc_{M_f}(x)=f(x).$ Let $q(|x|)$ be the number of non-deterministic choices of $M_f.$ Let $p=\frac{f(x)}{2^{q(|x|)}}$. We can compute an estimate $\hat{p}$ of $p,$ by choosing $m=poly(\epsilon^{-1}, \log \delta^{-1})$ paths uniformly at random. Then we can compute $\hat{f(x)}=\hat{p}2^{q(|x|)}$.

If $f(x)=0$ then $\hat{f(x)}=0.$ If $f(x)\neq 0$, then $p>\frac{1}{2},$ so by the   unbiased estimator lemma \ref{unbiased estimator}, $\hat{f(x)}$ satisfies the definition of f.p.r.a.s.

For the second inclusion, let $f\in FPRAS$, we will show that the decision version of $f$, i.e. deciding if $f(x)=0$, is in RP. On input $x$ we run the f.p.r.a.s. for $f$ with e.g. $\epsilon=\delta=\frac{1}{4}.$ We return yes iff $\hat{f(x)}\geq \frac{1}{2}.$ 

If $f(x)=0$ the the f.p.r.a.s. returns $0$, so we return yes with probability $0$. If $f(x)\geq 1$, then $\hat{f(x)}\geq \frac{1}{2}$ with probability at least $1-\delta$, so we return yes with the same probability.
\end{proof}

\begin{corollary}\label{Corollary1}
If $FPRAS \subseteq TotP$ then RP=P.
\end{corollary}
\begin{proof}
If FPRAS$\subseteq$ TotP, then \#RP1 $\subseteq$ TotP, and then for all $f \in \#RP1$ $L_f\in$ P. So if $A\in RP$ then $\#A\in \#RP1$, and thus $A=L_{\#A}\in$ P. Thus RP=P.
\end{proof}

\begin{corollary}\label{Corollary2}
If $\#RP1=\#RP2$ then NP=RP.
\end{corollary}
\begin{proof}
If $\#RP1=\#RP2$ then they are both equal to FPRAS, thus TotP $\subseteq$ FPRAS, thus NP=RP from theorem \ref{TotP-vs FPRAS 1}.
\end{proof}

\subsection{The four possible worlds}

Now we will explore further relationships between the above mentioned classes, and we will present all four possible worlds inside \#P, with respect to NP vs RP vs P.

\begin{theorem}\label{four worlds}
The following figures hold. 
 
In the figures, $A \rightarrow B$ means $A \subseteq B$,  $A\dashv B$ means $A \not\subset B$, and $A \mapsto B$ means $A \subsetneq B$.

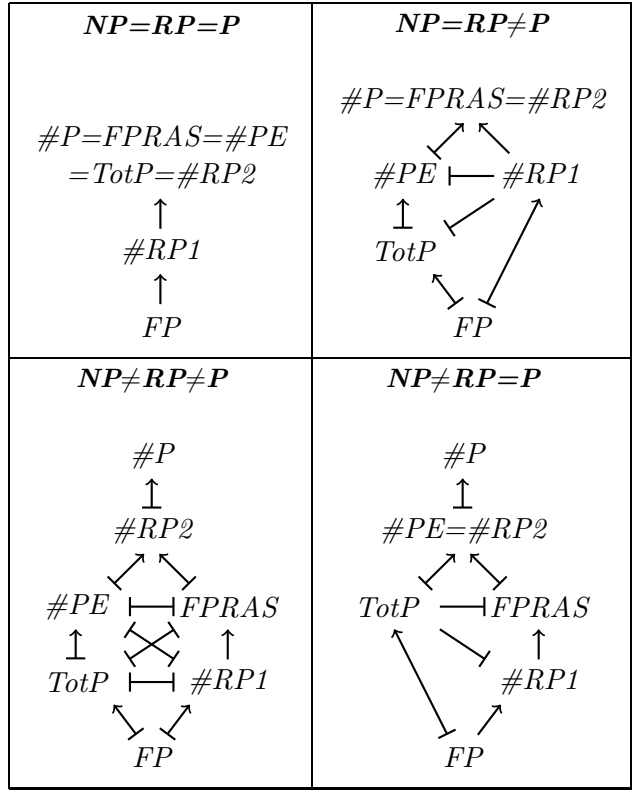
\begin{center}
\begin{tabular}{|c|c|}
\hline 
\begin{tikzpicture}
\node at (1,3) {\textbf{\small{NP=RP=P}}};
\node at (1,1.5) {\#P=FPRAS=\#PE};
\node at (1,1) {=TotP=\#RP2};
\draw [thick, ->] (1,0.3) -- (1,0.7);
\node at (1,0) {\#RP1};
\node at (1,-1) {FP};
\draw [thick, ->] (1,-0.7) -- (1,-0.3);
\end{tikzpicture} 
& 
\begin{tikzpicture}
\node at (1,3) {\textbf{\small{NP=RP$\neq$P}}};
\node at (1,2) {\#P=FPRAS=\#RP2};
\draw [thick, |->] (0.5,1.3 ) -- (0.9,1.7);
\node at (0.1,1) {\#PE};
\draw [thick, ->] (1.5,1.3) -- (1.1,1.7);
\node at (1.9,1) {\#RP1};
\draw [thick, |-] (0.7,1) -- (1.3,1);
\draw [thick, |->] (0.1,0.3) -- (0.1,0.7);
\node at (0.1,0) {TotP};
\draw [thick, -|] (1.3, 0.7) -- (0.7,0.3);
\node at (1,-1) {FP};
\draw [thick, |->] (0.8,-0.7) -- (0.5,-0.3);
\draw [thick, |->] (1.2,-0.7) -- (1.9,0.7);
\end{tikzpicture} 
\\ 
\hline 
\begin{tikzpicture}
\node at (1,4) {\textbf{\small{NP$\neq$RP$\neq$P}}};
\node at (1,3) {\#P};
\draw [thick, |-> ] (1,2.3) -- (1,2.7);
\node at (1,2) {\#RP2};
\draw [thick, |->] (0.5,1.3) -- (0.9,1.7);
\node at (0,1) {\#PE};
\draw [thick, |->] (1.5,1.3) -- (1.1,1.7);
\draw [thick, |-|] (0.7,1) -- (1.3,1);
\node at (2,1) {FPRAS};
\draw [thick, |->] (0,0.3) -- (0,0.7);
\node at (0,0) {TotP};
\draw [thick, |-|] (0.7,0) -- (1.3,0);
\draw [thick, ->] (2,0.3) -- (2,0.7);
\node at (2,0) {\#RP1};
\draw [thick, |-|] (0.7,0.3) -- (1.3,0.7);
\draw [thick, |-|] (1.3,0.3) -- (0.7,0.7);
\node at (1,-1) {FP};
\draw [thick, |->] (0.8,-0.7) -- (0.5,-0.3);
\draw [thick, |->] (1.2,-0.7) -- (1.5,-0.3);
\end{tikzpicture} 
& 
\begin{tikzpicture}
\node at (1,4) {\textbf{\small{NP$\neq$RP=P}}};
\node at (1,3) {\#P};
\draw [thick, |-> ] (1,2.3) -- (1,2.7);
\node at (1,2) {\#PE=\#RP2};
\draw [thick, |->] (0.5,1.3) -- (0.9,1.7);
\node at (0,1) {TotP};
\draw [thick, |->] (1.5,1.3) -- (1.1,1.7);
\draw [thick, -|] (0.7,1) -- (1.3,1);
\node at (2,1) {FPRAS};
\draw [thick, ->] (2,0.3) -- (2,0.7);
\node at (2,0) {\#RP1};
\draw [thick, |-] (1.3,0.3) -- (0.7,0.7);
\node at (1,-1) {FP};
\draw [thick, ->] (1.2,-0.7) -- (1.5,-0.3);
\draw [thick, |->] (0.8,-0.7) -- (0.1,0.7);
\end{tikzpicture} \\ 
\hline 
\end{tabular}
\end{center} 
\end{theorem}

\begin{proof}
First of all, intersections between any of the above classes are non-empty, because FP is a subclass of all of them.
For the rest inclusions, we have the following.

\#P$\subseteq$\#RP2 $\Leftrightarrow $ NP=RP (by definitions).

\#RP2 $\subseteq$ \#PE $\Leftrightarrow$ RP=P (by definitions).

\#P $\subseteq$ TotP $\Rightarrow$ NP=P (by theorem \ref{PZmt}).
Also NP=P $\Rightarrow$ \#SAT $\in$ TotP, because \#SAT is self-reducible, and it would have easy decision. Thus \#P $\subseteq$ TotP, because TotP is the Karp-closure of self reducible problems with easy decision, and \#SAT is \#P-complete under Karp reductions. So we get 
\#P = \#PE =   TotP $\Leftrightarrow$ P=NP.  

TotP $\subseteq$ FPRAS $\Rightarrow$ NP=RP (by theorem \ref{TotP-vs FPRAS 1}).

\#RP2 $\subseteq$ FPRAS $\Rightarrow$ TotP $\subseteq$ FPRAS $\Rightarrow$ NP=RP.

\#RP2=\#RP1 $\Rightarrow$ NP=RP (by corollary \ref{Corollary2}).

\#P $\subseteq$ FPRAS $\Leftrightarrow$ NP=RP (by theorem \ref{TotP-vs FPRAS 1}).

\#PE $\subseteq$ FPRAS $\Rightarrow$ TotP $\subseteq$ FPRAS $\Rightarrow$ NP=RP.

\#PE $\subseteq$ \#RP1 $\Rightarrow$ \#PE $\subseteq$ FPRAS  $\Rightarrow$ NP=RP.

\#RP1 $\subseteq$ \#PE $\Rightarrow$ RP=P (same proof as corollary \ref{Corollary1}).

FPRAS $\subseteq$ \#PE $\Rightarrow$ \#RP1 $\subseteq$ \#PE $\Rightarrow$ RP=P.

FPRAS $\subseteq$ TotP $\Rightarrow$ RP=P (corollary \ref{Corollary1}).

TotP $\subseteq$ \#RP1 $\Rightarrow$ TotP $\subseteq$ FPRAS $\Rightarrow$ NP=RP.

\#RP1 $\subseteq$ TotP $\Rightarrow$ RP=P (same proof as corollary \ref{Corollary1}).

\#RP1 $\subseteq$ FP $\Rightarrow$ \#RP1 $\subseteq$ TotP $\Rightarrow$ RP=P.

TotP $\subseteq$ FP $\Rightarrow$ NP=P (by theorem \ref{PZmt}).
\end{proof}

\section{Discussion and Further Research}\label{discussion}
Regarding the question of whether  FPRAS $\subseteq$ TotP, Corollary 1 means that it does not hold unless P=RP, and if it holds, then it is more difficult to prove it, than proving RP=P. 

 Determining if RP=P implies FPRAS $\subseteq$ TotP, is an  open question, interesting on its own. 
Morover, if it holds, then the (FPRAS $\subseteq$ TotP?) question would be an equivalent formulation of the RP vs P question. Note that we already know that the (TotP $\subseteq$ FPRAS?) question is an equivalent formulation of NP vs RP.

Corollary 2 means that although \#RP1 and \#RP2 are natural counting counterparts of RP, they can't be equal unless NP=RP. This happens because when we consider functions in \#P as functions that count the number of solutions to some decision problem, then in the first class the accepting random choices are in 1-1 correspondence with the solutions we want to count, while in the second class, random choices and solutions do not necessarily coincide. 

It is interesting that regarding decision, the two analogue  definitions of RP are equivalent, i.e. for every problem in RP, there exists an NPTM in ${\cal MR}$ deciding the problem, and vice versa. But for the  counting versions, this can't hold any more, unless NP=RP.

Regarding the possibility of characterizing FPRAS in terms of counting paths of Turing machines, Corollary 2 implies such a characterization, if the two classes coincide. However, in this case, as we saw, FPRAS would also be equal to \#P, which of course  is a much simpler characterization. 

On the other hand, without any assumptions, theorem \ref{main theorem} means that the characterization of FPRAS that we are looking for, is something between the two definitions of \#RP1 and \#RP2, a relaxation of the first and a restriction of the second.  However it might be not very elegant.

It is an open question whether FPRAS is in \#RP1. It seems that both a negative and a positive answer are compatible with all four worlds. If the answer is positive, then \#RP1 would be an equivalent characterization of FPRAS.

\section{Related work}\label{related work}
There is an enormous literature on approximation algorithms and inapproximability results for individual problems in \#P, e.g.  \cite{Valiant79,JS96perm,KLM89,DFJ02,Wei06}.

From a unifying point of view, the most important results regarding approximability are the following. 
Every function in \#P either admits a f.p.r.a.s., or is inapproximable to within a polynomial factor \cite{JS89}.
For self-reducible problems in \#P, f.p.r.a.s. is equivalent to almost uniform sampling \cite{JS89}.
With respect to approximation preserving reductions, there are three main classes of functions in \#P: (a) the class of functions  that admit an f.p.r.a.s., (b) the class of fuctions that are interreducible with \#SAT under approximation preserving reductions, and  (c) the class of problems that are interreducible with \#BIS under approximation preserving reductions \cite{DGGJ04}. Problems in the second class are inapproximable unless NP=RP, while the approximability status of the third class is unknown.

Several classes of problems with easy decision version are studied in   \cite{KPSZ01,PZ06,Pagourtzis01,BGPT17,HHKW07, AJ93}. Also counting problems in \#P are studied in terms of  descriptive complexity in \cite{Saluja, Bulatov, Arenas, DGGJ04, Dalmau}. 
  
Classes that are so far proven to admit f.p.r.a.s.  are \#R$\Sigma$2, \#$\Sigma$1 \cite{Saluja} and some extensions of \#$\Sigma$1 \cite{Arenas}. Note that all of them are subclasses of TotP.

\section{Conclusions} 
In this paper we refined the picture of counting complexity inside \#P, in particular with regard to approximate counting.
 
We have introduced two subclasses of \#P: \#RP1 and \#RP2, that both can be considered as counting versions of RP. We showed that FPRAS lies between them, and that these classes do not coincide unless NP=RP. With the help of these classes we also proved that FPRAS cannot be contained in TotP, unless RP=P. 

Finally we studied the implications of randomized decision complexity, i.e. RP vs NP, and RP vs P, to the relationships between some counting classes inside \#P, that have decision version either in P or in RP. 

Further research could be the study of the relationships  between the new introduced classes with other known subclasses of \#P, as mentioned in the related work section \ref{related work}.  
\section*{Appendix}

\begin{theorem}\label{theorem2} (folklore)
If NP=RP then all problems in \#P admit an f.p.r.a.s. 
\end{theorem}
\begin{proof} (sketch)
In \cite{Stockmeyer85a} Stockmeyer has proven that an f.p.r.a.s., with access to a $\Sigma_2^p$ oracle, exists for any problem in \#P. If NP=RP then $\Sigma_2^p= RP^{RP}\subseteq BPP$. Finally it is easy to see that an f.p.r.a.s. with access to a BPP oracle, can be replaced by another  f.p.r.a.s., that simulates the oracle calls itself.
\end{proof}


\end{document}